\documentclass[pra,twocolumn,showkeywords,longbibliography,superscriptaddress,10pt]{revtex4-1}
\usepackage{graphicx}
\usepackage{dcolumn}
\usepackage{color}
\usepackage{times}
\usepackage{bm}
\usepackage{amssymb}
\usepackage{amsmath}
\usepackage{epsfig}
\usepackage{epstopdf}
\usepackage{dsfont}
\usepackage{subfigure}
\usepackage{tikz}
\usepackage[colorlinks, citecolor=blue, linkcolor=blue,urlcolor=blue]{hyperref}
\usepackage[mathscr]{euscript}
\usepackage{amsthm}
\usepackage{xcolor}

\newtheoremstyle{myprop}
  {\topsep}              
  {\topsep}              
  {\upshape}             
  {}                     
  {\scshape\itshape}     
  {.}                    
  { }                    
  {\thmname{#1}\thmnumber{ #2}} 
\theoremstyle{myprop}
\newtheorem{proposition}{Proposition} 

\newcommand{\Tr}{\operatorname{Tr}}
\newcommand{\norm}[1]{\left\| #1 \right\|}
\newcommand{\Fnorm}[1]{\norm{#1}_F}

\makeatletter
\renewcommand{\maketag@@@}[1]{\hbox{\m@th\normalsize\normalfont#1}}
\makeatother

\begin{document}
\renewcommand{\thefootnote}{\fnsymbol {footnote}}

\title{\textcolor{black}{Imprecise quantum steering inequalities in tripartite systems}}

\author{Yan Zhao}
\affiliation{School of Physics, Anhui University, Hefei
    230601,  People's Republic of China}

\author{Li-Juan Li}
\affiliation{School of Physics, Anhui University, Hefei 230601, People's Republic of China}

\author{Zheng-Peng Xu}
\affiliation{School of Physics, Anhui University, Hefei 230601,  People's Republic of China}

\author{Liu Ye}
\affiliation{School of Physics, Anhui University, Hefei 230601,  People's Republic of China}

\author{Dong Wang} \email{dwang@ahu.edu.cn}
\affiliation{School of Physics, Anhui University, Hefei
    230601,  People's Republic of China}

\date{\today}

\begin{abstract} Quantum steering, as a manifestation of nonlocal quantum correlations, plays a crucial role in enabling various quantum information processing tasks. However, practical implementations are often hindered by significant challenges arising from imperfect or untrusted measurement devices. This study investigates the impact of measurement inaccuracies on quantum steering, with a particular focus on errors in the untrusted party’s measurement devices. We first analyze how such errors affect the evaluation of steering inequalities, and then derive bipartite steering inequalities based on correlation matrices under imperfect measurements.
    Our findings show that even small measurement errors can significantly compromise the certification of quantum steerability, an effect that becomes particularly pronounced as the system dimension increases. Furthermore, by extending the proposed steering inequality to a modified tripartite scenario via correlation matrices, we demonstrate that the influence of measurement imperfections is far more severe in multipartite quantum steering than in the bipartite case. Our results underscore the critical need to account for measurement imperfections in experimental quantum steering and provide a theoretical framework for characterizing and mitigating these effects in high-dimensional quantum systems.
\end{abstract}

\maketitle

\section{Introduction}
{Quantum steering stands out as an asymmetric form of nonlocality, where one party (typically referred to as Alice) can remotely influence the spatially separated quantum state of another party (Bob) through her choice of measurements. This concept was first introduced by Schr\"{o}dinger in 1935 \cite{Schrodinger_1935} in response to Einstein, Podolsky, and Rosen's challenge to the completeness of quantum mechanics--the celebrated EPR paradox \cite{PhysRev.47.777}, and was rigorously formalized in 2007 in terms of local hidden state (LHS) models through the pioneering work of Wiseman and his colleagues \cite{PhysRevLett.98.140402,PhysRevA.76.052116}. It has become a cornerstone in the study of quantum correlations and their applications in quantum technologies. }

As a kind of quantum nonlocality, EPR steering is different from quantum entanglement and Bell nonlocality \cite{GUHNE20091,RevModPhys.86.419}. EPR steering is proven to be inherently asymmetric with respect to the two parties of the system: For certain states, one party may steer the other, but not vice versa \cite{PhysRevLett.112.200402,PhysRevA.93.022121,Zhong:17,PhysRevA.81.022101,PhysRevA.88.051802,PRXQuantum.3.030102}. Unlike the standard Bell scenario, Alice's devices are fully characterized (the ``trusted" side), while Bob's devices remain uncharacterized (a ``black box"). This  quantum steering is experimentally more demanding than entanglement witnessing yet notably less stringent than Bell nonlocality \cite{Saunders2010,PhysRevLett.116.160404,Armstrong2015,Handchen2012}. While violations of traditional Bell inequalities vanish in many-body large-scale systems, the steering inequality proposed can maintain unbounded violations, demonstrating that quantum steering surpasses Bell nonlocality in both experimental feasibility and noise robustness \cite{PhysRevLett.118.020402}. The trust in one party's devices grants steering superior robustness to noise and detector inefficiencies properties that scale favorably with increasing Hilbert space dimensions. Consequently, high-dimensional quantum steering experiments have attracted substantial interest in recent years, with experimental capabilities now extending to systems of dozens of dimensions \cite{PhysRevLett.120.030401,doi:10.1126/science.aar7053,PhysRevX.12.041023,PhysRevLett.126.200404,Qu:22,PhysRevLett.127.020401,PhysRevLett.128.240402}. This growing focus is due to the recognition of quantum steering as a vital and effective quantum resource. It plays a pivotal role in various quantum information processing tasks, including one-sided device-independent quantum key distribution \cite{PhysRevA.85.010301,Gehring2015,Walk:16}, secure quantum teleportation \cite{PhysRevLett.115.180502,PhysRevA.88.062338}, quantum randomness certification \cite{PhysRevLett.120.260401,Passaro_2015}, and subchannel discrimination \cite{PhysRevLett.114.060404,Sun2018}. These applications highlight the practical importance of quantum steering in enhancing the security and efficiency of quantum communication and computation.

A central challenge in the study of EPR steering lies in determining the steerability of quantum states. Therefore, researchers have established diverse criteria  for identifying steerable states, spanning theoretical frameworks, experimental protocols, and operational characterizations. These criterias based on linear steering inequalities \cite{PhysRevA.80.032112,PhysRevA.95.012142,Saunders2010}, the local uncertainty relation \cite{PhysRevLett.106.130402,PhysRevA.87.062103,PhysRevA.92.062130,PhysRevA.93.012108,PhysRevA.98.050104,PhysRevA.98.062111}, the all-versus-nothing proof \cite{Wu2014}, and Clauser-Horne-Shimony-Holt-like inequalities \cite{PhysRevA.94.032317,PhysRevA.95.062111,Cavalcanti:15}. However, quantum steering scenarios also rely on the assumption that one party must perform perfectly characterized measurements. This requirement introduces practical challenges, as real-world experiments often involve measurement devices that are not entirely precise or trustworthy,  inherently suffering from noise, finite precision, and calibration errors. For example, photon detectors exhibit dark counts (false signals from thermal noise) \cite{10.1117/12.2017357} and inefficiencies (missed photon detections); optical components (e.g., waveplates, beam splitters) have alignment inaccuracies, leading to deviations from idealized projective measurements; temperature fluctuations or electromagnetic interference introduce stochastic errors.
Thus, if errors exist in the trusted measurement devices, non-steerable states may be falsely identified as steerable. For instance, deviations in Bob's measurement basis from the theoretical design could lead to the observation of spurious steerability, this poses a catastrophic negative impact on the applications of quantum steering. Moreover, this lack of trust exhibits dimension-dependent sensitivity: for a $d$-dimensional system, the error impact scales as $\mathcal{O}(d^3)$ due to the growing Hilbert space \cite{PhysRevA.108.L040401}. For example, a 1\% error in a 10-dimensional system may require more than $99\%$ trust to avoid misclassifications \cite{PhysRevA.108.L040401}. This dimension-dependent sensitivity necessitates both near-perfect device characterization and the use minimal measurements \cite{sarkar2024deviceindependentcertificationgmestates} to achieve reliable steering verification in scalable quantum architectures.

In this paper, we investigate the impact of measurement inaccuracies in Bob's side on quantum steering scenarios. We firstly examine how errors affect measurement outcomes. We model errors by quantifying deviations between target and implemented measurement operators via the square of Frobenius norm, showing their impact grows with system dimension. Then, we derive steering inequalities for correlation matrices under inaccurate measurements, finding that even small errors can significantly influence the results of steering inequalities, and this effect becomes more pronounced as the dimension increases. Finally, we extend the current analysis to tripartite correlation matrices, discovering that in tripartite steering, the impact of these errors is even more substantial than in bipartite steering.

This paper is organized as follows. Section \ref{Sec2}   outlines foundational concepts of quantum steering,  and Section \ref{Sec3} reviews established criteria for detecting steerable states. In Section \ref{Sec4}, we propose a generalized steering criterion tailored for high-dimensional bipartite systems and provide a rigorous proof of its validity. Additionally, we analyze the interplay between measurement imperfections and steering robustness, deriving error bounds for practical implementations. Finally, we conclude this paper with a concise discussion and conclusion.

\section{Einstein-Podolsky-Rosen steering}\label{Sec2}
{The EPR steering task in quantum nonlocality verification is typically described through the following protocol framework: consider such an entangled system $\hat{\rho}_{AB}$, which consists of two spatially separated parties, Alice and Bob.
Alice initially prepares a bipartite entangled state $\hat{\rho}_{AB}$ of local dimension $d$, with reduced density matrices satisfying $ {\hat{\rho}_a}= \mathrm{{Tr}_b}(\hat{\rho}_{AB})$ and ${\hat{\rho}_b} = \mathrm{{Tr}_a}(\hat{\rho}_{AB})$, then transmits the subsystem to Bob. Alice's core objective resides in enabling Bob who exclusively trusts his own measurement apparatus to become convinced of the existence of quantum entanglement between their systems. To achieve this, Alice must induce observable steering effects on Bob's quantum state through measurement operations.
Specifically, Alice and Bob independently carry out a series of measurement with setting $x$ and $y$ on their respective subsystems and obtained the corresponding outcomes $a, b\in\{1, ...,d\}$. Alice's measurement operators \(\{\hat{A}_{a|x}\}\) remain \textit{uncharacterized}, while Bob's measurement operators \(\{\hat{B}_{b|y}\}\) are fully specified as known quantum projective measurements. Then, the correlations can be derived by applying the Born rule
\textcolor{black}{\begin{align}
        p(a,b|x,y) = \mathrm{Tr}\left[\left(\hat{A}_{a|x} \otimes \hat{B}_{b|y}\right)\hat{\rho}_{AB}\right].
        \label{Eq.1}
    \end{align}}
A quantum system exhibits steering phenomena if its correlations cannot be described by  \textit{local hidden state} (LHS) model. Specifically, it can be mathematically expressed as:

\begin{equation}
    p(a,b|x,y) =\sum_{\lambda} \, p(\lambda) \, p(a|x,\lambda) \, \mathrm{Tr}\left( \hat{B}_{b|y} \, \hat{\sigma}_{\lambda} \right),
    \label{Eq.2}
\end{equation}
where $\lambda$ denotes the predetermined hidden-variable parameter, $p(a|x, \hat{\lambda})$ represents the local response function, and $\sigma_\lambda$ is the hidden state. Bob's unnormalized conditional state $\hat{\sigma}_{a|x}=\mathrm{Tr}_a\left[ \left( \hat{A}_{a|x} \otimes \mathbb{I} \right){\hat{\rho}_{AB}} \right]$ corresponding to measurement setting \( x \) and outcome \( a \) chosen by Alice satisfies $\sum_{a} \sigma_{a|x} = \rho_B$ ensuring Alice's choice of measurement does not affect Bob's reduced state. Bob subsequently verifies whether his assemblage of conditional states $\{\sigma_{a|x}\}$can be described by a LHS model
\begin{equation} \label{Eq.3}
    \hat{\sigma}_{a|x} = \sum_{\lambda} p(\lambda) \, p(a|x,\lambda) \, \hat{\sigma}_\lambda.
\end{equation}
The shared bipartite state $\hat{\rho}_{AB}$ is steerable if the assemblages don't admit such an LHS model.

{

In Ref. \cite{PhysRevA.106.042402}, the authors derive an efficient linear criterion for arbitrary bipartite system, which is based on correlation matrices, the following definitions and conclusions are quoted from their work. Consider a bipartite quantum state \(\rho\) shared between Alice and Bob on a Hilbert space \(\mathcal{H} = \mathcal{H}_a \otimes \mathcal{H}_b\), while $\hat{\rho}_a$ and $\hat{\rho}_b$ are reduced states for Alice and Bob respectively. Moreover, two sets of local observables of Alice and Bob are labeled by $\mathcal{A} = \{ \hat{A}_i \mid i = 1, 2, \ldots, m \}$ and $\mathcal{B} = \{\hat{B}_j \mid j = 1, 2, \ldots, m \}$. Then, this allows us to construct the correlation matrix associated with these measurements as
\begin{equation} \label{eq:correlation_matrix}
    C(\mathcal{A}, \mathcal{B} \mid \rho) = \bigl( c_{ij} \bigr),
\end{equation}
with  each entry being $
    c_{ij} = \operatorname{Tr}\left[ (\hat{A}_i \otimes \hat{B}_j)(\hat{\rho} - \hat{\rho}_a \otimes \hat{\rho}_b) \right].
$

It has been proven that, if the quantum state $\rho$ is unsteerable from Alice to Bob,  then
\begin{equation} \label{eq:steer inequality}
    \|C(\mathcal{\hat{A}}, \mathcal{\hat{B}} \mid \hat{\rho})\|_{\operatorname{Tr}} \leqslant \sqrt{\Lambda_a \Lambda_b},
\end{equation}
where the parameters are defined as:
\begin{equation}
    \Lambda_a = \sum_{i=1}^{m} V(\hat{A}_i, \hat{\rho}_a), \label{eq:Lambda_a}
\end{equation}
\begin{equation}
    \Lambda_b = \max_{\sigma_b} \left( \sum_{j=1}^{n} \left( \operatorname{Tr}(\hat{B}_j \hat{\sigma}_b) \right)^2 \right) - \sum_{j=1}^{n} \left( \operatorname{Tr}(\hat{B}_j \hat{\rho}_b) \right)^2, \label{eq:Lambda_b}
\end{equation}
here, $\|C(\mathcal{\hat{A}}, \mathcal{\hat{B}} \mid \hat{\rho})\|_{\operatorname{Tr}}$ denotes the trace norm of the matrix $C(\mathcal{\hat{A}}, \mathcal{\hat{B}} \mid \hat{\rho})$, which is defined as the sum of its singular values, and ${V}(\hat{A}_i, \hat{\rho}_a)=\operatorname{Tr}(\hat{A}_i^2 \hat{\rho}_a) - \left( \operatorname{Tr}(\hat{A}_i \hat{\rho}_a) \right)^2$ represents the conventional variance  of $A_i$ in the state $\hat{\rho}_a$ and the maximum is over all states $\hat{\sigma}_b$ on Bob's side.

In particular, for the special case where the local observables $\{\hat{A}_i\}$ and $\{\hat{B}_j\}$ are chosen to form local orthonormal operator (LOO) bases, the parameters $\Lambda_a$ and $\Lambda_b$ in (5) simplify to the following dimension dependent expressions \cite{PhysRevA.106.042402}:
\[
    \Lambda_a = d_1 - \operatorname{Tr}(\hat{\rho}_a^2), \qquad
    \Lambda_b = 1 - \operatorname{Tr}(\hat{\rho}_b^2),
\]
where $d_1$ is the dimension of Alice’s subsystem.

Furthermore, it also establishes a steering criterion based on the trace norm of the correlation matrix. Let \( m = n \), suppose the state \( \hat\rho \) is unsteerable from Alice to Bob, then the following inequality holds for any real numbers \( g_i \)
\begin{equation} \label{eq:trace norm inequality}
    \sum_{i=1}^m \left| \operatorname{Tr}\left( (g_i \hat{A}_i \otimes \hat{B}_i)(\hat{\rho} - \hat{\rho}_a \otimes \hat{\rho}_b) \right) \right| \leq \sqrt{\Lambda_a({g}) \Lambda_b},
\end{equation}
where \( \Lambda_b \) has been defined by Eq.~(\ref{eq:Lambda_b}), and the  conventional variance is given as fallows
\begin{equation} \label{eq:Lambda_a_g}
    \Lambda_a(\bm{g}) = \sum_{i=1}^m g_i^2 V(\hat{A}_i, \hat{\rho}_a).
\end{equation}
The freedom in choosing the parameters \(g_i\) in inequality  affords the flexibility to optimize the degree of its violation. Although these parameters are not intrinsically required, they serve a practical purpose by allowing the detection power of the steering inequality to be enhanced through their adjustment alone. As an example, given specific observables for party Alice and fixed observables for Bob, optimal measurement settings for Alice can be easily derived by simply selecting \(g_i = 0\) or \(1\).

\section{Imprecise measurements}\label{Sec3}
In the aforementioned steering criteria, we assume that Bob's measurement outcomes are trusted, which is a standard requirement in steering protocols. Now let us address scenarios where Bob's measurement devices are not fully trusted. This extension is critical for semi-device-independent steering frameworks, where partial trust in Bob's apparatus relaxes experimental constraints while maintaining security guarantees.

To quantify the errors occurring in each of Bob's measurements, consider a local quantum state in a $d$ dimensional space and a set of tomographically complete observables \(\{\hat{\sigma}_i\}\) (with \(i=1,\ldots,d^2\) and \(\sigma_1 = \mathbb{I}\)) that his device is expected to implement. Further, let \(\{\hat{\tau}_i\}\) denote the observables actually performed by the device (assumed unitary for simplicity).
These observables form a LOO basis, satisfying $\operatorname{Tr}(\hat{\sigma}_i \hat{\sigma}_j^\dagger) = 0$ for all $i \neq j$, and $\operatorname{Tr}(\hat{\rho} \hat{\sigma}_i) \leq 1$ for for all $i$. The error quantification \(\Gamma_B\) of Bob's measurement apparatus can be characterized based on the deviation between the target operators \(\{\hat{\sigma}_i\}\) and the implemented operators \(\{\hat{\tau}_i\}\), which had been studied in Refs. \cite{PhysRevA.108.L040401, PhysRevLett.132.070204, PhysRevLett.133.150201}. Fidelity  places more emphasis on the overall similarity of states \cite{PhysRevLett.132.070204}, whereas the squared Frobenius norm directly captures the implementation deviation of the operators themselves, making it more convenient for integration into operator-expansion-based theoretical derivations. Following the approach in \cite{PhysRevA.108.L040401}, we quantify the error using the squared Frobenius norm between the intended and real measurement operators. This parameter reflects the discrepancy between experimental operations and theoretical expectations
\begin{equation} \label{eq:Gamma_B}
    \Gamma_B = \left\| \hat{\sigma}_i - \hat{\tau}_i \right\|,
\end{equation}
where $||A||=\mathrm{Tr}(AA^\dagger)$. This framework explicitly bridges theoretical assumptions and experimental deviations, establishing a mathematical foundation for error analysis in semi-device-independent protocols.

We now rigorously analyze the effects of experimentally realistic imperfections on the operational validity of steering inequalities in quantum protocols.
Since \(\{\hat{\sigma}_i\}\) forms a complete basis for matrices acting on \(\mathbb{C}^d\), we can express any density matrix \(\rho_b\) and its inferred version \(\rho_b^{\mathrm{inf}}\) as
\[
    \rho_b = \sum_i r_i \hat{\sigma}_i, \quad \rho_b^{\mathrm{inf}} = \sum_i q_i \hat{\sigma}_i,
\]
where the coefficients \(r_i, q_i \in \mathbb{R}\) (due to the Hermiticity of \(\rho_a\)) satisfy \(|r_i|, |q_i| \leq 1\), and $\rho_{b}^{\mathrm{inf}}$ is the inferred state from Bob's imprecise tomography.

Our goal is to bound \(|q_i - r_i|\) for all \(i\). When Bob implements measurements \(\{\hat{\tau}_i\}\) instead of \(\{\hat{\sigma}_i\}\), the coefficients \(q_i\) and \(r_i\) are determined by
\begin{equation} \label{eq:q_definition}
    q_i = \frac{1}{d} \operatorname{Tr}\left( \hat{\rho}_b \hat{\tau}_i^\dagger \right),
    r_i = \frac{1}{d} \operatorname{Tr}\left( \hat{\rho}_b \hat{\sigma}_i^\dagger \right).
\end{equation}
Substituting the expansion \(\rho_b = \sum_j r_j \hat{\sigma}_j\) into \eqref{eq:q_definition} yields
\begin{equation} \label{eq:q_expansion}
    q_i = \frac{1}{d} \sum_j r_j \operatorname{Tr}\left( \hat{\sigma}_j \hat{\tau}_i^\dagger \right).
\end{equation}
In Ref. \cite{PhysRevA.108.L040401}, the distance between these traces while $ \left\|{\hat\tau}_i - \hat{\sigma}_i \right\|\leq\xi$ is given by
\begin{align} \label{eq:error_bound}
    |{r}_i - {q}_i| \leq d \left( \frac{\xi}{2} + \sqrt{2d\xi} \right).
\end{align}

\section{steering  inequality with Imprecise measurements}\label{Sec4}
{We now examine how experimental imperfections affect the robustness of steering inequalities in quantum protocols. Specifically, we quantify the deviation of the steering inequality mentioned before under small but finite perturbations of the target measurement operators, thereby characterizing the stability of quantum steering detection in nonideal implementations.}

\begin{proposition}
    Consider a bipartite state $\hat{\rho}_{AB}$ where Alice performs uncharacterized measurements ${\hat{A}_i}$, while Bob attempts to perform a trusted set of projective measurements ${\hat{B}_i}$ but instead implements imperfect observables ${\tilde{\hat{B}}_i}$. Assume that the deviation between each target and implemented observable is bounded  as $||{\tilde{\hat{B}}_i - \hat{B}_i}|| \leq \xi$, and let $\hat{\rho}_{\text{diff}}=(\hat{\rho} - \hat{\rho}_a \otimes \hat{\rho}_b)$. Then the following inequality holds:
    \begin{equation}\label{proof10}
        \sum_{i=1}^{m}|\mathrm{Tr}[(\hat{A}_{i}\otimes \hat{B}_{i})\hat{\rho}_{\text{diff}}]|-\sqrt{\xi} \sum_{i=1}^{m} \sqrt{\Tr(\hat{A}_i\hat{A}_i^\dagger)}\leq \sqrt{\Lambda_a \tilde{\Lambda}_b},
    \end{equation}
    where
    $\tilde{\Lambda}_{b}\leq\Lambda_{b}+2\sum_{i=1}^{m}\eta[1+\mathrm{Tr}(\hat{B}_{i}\hat{\rho}_{b})]$ holds, with $\eta=d^2\left( \frac{\xi}{2} + \sqrt{2d\xi} \right)$, and $d$ is Bob's local dimension.
\end{proposition}

\begin{proof}
    Due to imperfections in Bob's measurement implementations, he performs $\{\tilde{\hat{B}}_i\}$ instead of $\{\hat{B}_i\}$, the corresponding steering inequality based on the trace of correlation matrix takes the form
    \begin{equation} \label{eq:erorr steering_inequality}
        \sum_{i=1}^m \left| \operatorname{Tr}\left( (\hat{A}_i \otimes \tilde{\hat{B}}_i)\hat{\rho}_{\text{diff}} \right) \right| \leq \sqrt{\Lambda_a \tilde{\Lambda}_b},
    \end{equation}
    where,
    \begin{equation} \label{eq:Lambda_a}
        \Lambda_a = \sum_{i=1}^m V(\hat{A}_i, \hat{\rho}_a),
    \end{equation}
    \begin{equation}\label{eq:Lambda_a erorr}
        \tilde{\Lambda}_b = \max_{\sigma_b} \left(\sum_{j=1}^m (\mathrm{Tr}  (\tilde{\hat{B}}_j \hat{\sigma}_b))^2 - \sum_{j=1}^m (\mathrm{Tr} (\tilde{\hat{B}}_j \hat{\rho}_b))^2 \right).
    \end{equation}

    {For simplicity}, we set ${g_i}=1$  in this context. Let us first focus on the right-hand side of the inequality (\ref{eq:erorr steering_inequality}). Since Alice's measurements are uncharacterized, her conventional variance remains constant, which implies we need only consider $\tilde{\Lambda}_b$. By applying the triangle inequality, we can straightforwardly derive the following relationship
    \begin{equation}\label{inequality a}
        |\mathrm{Tr}( \hat{\sigma} _ {b} \tilde {\hat{B} } _ {j} ) | \leq |\mathrm{Tr}( \hat{\sigma}_{b} \hat{B}_{j} ) | + |\mathrm{Tr}( \hat{\sigma} _ {b} \hat{B} _ { j } ) -\mathrm{Tr}( \hat{\sigma} _ {b} \tilde { \hat{B} } _ { j } ) |,
    \end{equation}
    \begin{equation}\label{inequality b}
        |\mathrm{Tr}( \hat{\sigma} _ {b} \tilde {\hat{B} } _ {j} ) | \geq |\mathrm{Tr}( \hat{\sigma}_{b} \hat{B}_{j} ) | - |\mathrm{Tr}( \hat{\sigma} _ { b } \hat{B} _ { j } ) -\mathrm{Tr}( \hat{\sigma} _ {b} \tilde { \hat{B} } _ { j } ) |.
    \end{equation}
    From  Eq. (\ref{eq:error_bound}), we establish that
    \begin{equation}
        \left| \mathrm{Tr}(\hat{\rho}_b \tilde{\hat{B}}_j) - \mathrm{Tr}(\hat{\rho}_b \hat{B}_j) \right|\leq \eta,
    \end{equation}
    where $\eta=d^2\left( \frac{\xi}{2} + \sqrt{2d\xi} \right)$. By combining Eqs. \eqref{inequality a} and~\eqref{inequality b} and substituting them into Eq. \eqref{eq:Lambda_a erorr}, we obtain:
    \begin{equation}
        \begin{aligned}
            \tilde{\Lambda}_{b} \leq & \sum_{i=1}^{m} \left[ \mathrm{Tr}(\hat{B}_{i}\hat{\sigma}_{b}) +\eta \right]^{2}- \sum_{i=1}^{m} \left[ \mathrm{Tr}(\hat{B}_{i}\hat\rho_{b}) -\eta \right]^{2}.
        \end{aligned}
        \label{eq:Lambda_b_result}
    \end{equation}
    Through   simplification, we arrive at
    \begin{equation}
        \tilde{\Lambda}_{b}\leq\Lambda_{b}+2\sum_{i=1}^{m}\eta[\mathrm{Tr}(\hat{B}_{i}\hat{\sigma}_{b})+\mathrm{Tr}(\hat{B}_{i}\hat{\rho}_{b})].
    \end{equation}
    Apply H\"{o}lder's inequality for matrix norms given by
    \begin{equation}
        |\mathrm{Tr}(\hat{A}\hat{B}^\dagger)|\leqslant||\hat{A}||_F||\hat{B}||_F,
    \end{equation}
    where$||A||_F=\sqrt{||A||}$, we can obtain the following result
    \begin{equation}
        \mathrm{Tr}(\hat{B}_{i}\hat{\sigma}_{b})\leq \sqrt{\mathrm{Tr}(\hat{B}_{i}\hat{B}_{i}^{\dagger})\mathrm{Tr}(\hat{\sigma}_{b}\hat{\sigma}_{b}^{\dagger})}.
    \end{equation}
    Since \(\hat{\sigma}_b\) is a density operator satisfying \(\operatorname{Tr}(\hat{\sigma}_b) = 1\), the right-hand side of the inequality simplifies to \(\sqrt{\operatorname{Tr}(\hat{B}_i\hat{B}_i^\dagger)}\). Therefore, the right-hand side of the inequality reduces to
    \begin{equation}\label{proof11}
        \tilde{\Lambda}_{b}\leq\Lambda_{b}+2\sum_{i=1}^{m}\eta[1+\mathrm{Tr}(\hat{B}_{i}\hat{\rho}_{b})].
    \end{equation}

    By means of the same method, the left-hand side of the inequality can be derived using the triangle inequality. For each measurement setting \(i\), we have
    \begin{equation}
        \begin{split}
            \bigl| \Tr\bigl[ (\hat{A}_i \otimes \tilde{\hat{B}}_i) \, \hat{\rho}_{\mathrm{diff}} \bigr] \bigr|
            &\geq \bigl| \Tr\bigl[ (\hat{A}_i \otimes \hat{B}_i) \, \hat{\rho}_{\mathrm{diff}} \bigr] \bigr| \\
            &\quad - \bigl| \Tr\bigl[ (\hat{A}_i \otimes \Delta \hat{B}_i) \, \hat{\rho}_{\mathrm{diff}} \bigr] \bigr|,
        \end{split}
    \end{equation}
    where \(\Delta \hat{B}_i = \tilde{\hat{B}}_i - \hat{B}_i\) denotes the deviation between the implemented and the intended observables. \noindent The error term can be bounded further using H\"older's inequality for the Frobenius norm:
    \begin{equation}
        \left| \Tr\left[ (\hat{A}_i \otimes \Delta \hat{B}_i) \, \hat{\rho}_{\mathrm{diff}} \right] \right|
        \leq \Fnorm{\hat{A}_i} \; \Fnorm{\Delta \hat{B}_i} \; \Fnorm{\hat{\rho}_{\mathrm{diff}}}.
    \end{equation}

    We now turn to discuss each factor individually:
    from the error model, the deviation of the implemented observable satisfies \(\Fnorm{\Delta \hat{B}_i} \leq \sqrt{\xi}\). The Frobenius norm of \(\hat{A}_i\) is \(\Fnorm{\hat{A}_i} = \sqrt{\Tr(\hat{A}_i\hat{A}_i^\dagger)}\). For any bipartite state, the Frobenius norm of the difference matrix is bounded by \(\Fnorm{\hat{\rho}_{\mathrm{diff}}} \leq 1\) (a standard result for the deviation between a state and the product of its marginals). Substituting these results into the inequality yields
    \begin{equation}
        \left| \Tr\left[ (\hat{A}_i \otimes \Delta \hat{B}_i) \, \hat{\rho}_{\mathrm{diff}} \right] \right|
        \leq \sqrt{\xi} \; \sqrt{\Tr(\hat{A}_i\hat{A}_i^\dagger)}.
    \end{equation}

    \noindent Finally, summing over all measurement settings \(i = 1, \dots, m\), we obtain
    \begin{equation}\label{proof12}
        \begin{split}
            \sum_{i=1}^{m} \left| \Tr\left[ (\hat{A}_i \otimes \hat{B}_i) \, \hat{\rho}_{\mathrm{diff}} \right] \right|
            &- \sqrt{\xi} \sum_{i=1}^{m} \sqrt{\Tr(\hat{A}_i\hat{A}_i^\dagger)}
            \leq \\
            &\sum_{i=1}^{m} \left| \Tr\left[ (\hat{A}_i \otimes \tilde{\hat{B}}_i) \, \hat{\rho}_{\mathrm{diff}} \right] \right|.
        \end{split}
    \end{equation}
    By resorting to Eqs. \eqref{proof11} and \eqref{proof12}, we demonstrate that Eq. \eqref{proof10} has been proved.
\end{proof}

Next, the measurements are performed using two sets of local orthogonal bases $\{\mathcal{G}_i\}$ and $\{\mathcal{H}_j\}$, defined on the Hilbert spaces $\mathbb{C}^{d_1}$ and $\mathbb{C}^{d_2}$, respectively, where $d_1$ and $d_2$ are the local dimensions of the subsystem with $d_1 = d_2 = d$, and satisfy $\operatorname{Tr}(\mathcal{G}_i \mathcal{G}_j)=\operatorname{Tr}(\mathcal{H}_i \mathcal{H}_j)  = \delta_{ij}$. The correlation matrix constructed with the LOOs reads
\begin{equation}
    C(\mathcal{\hat{G}}, \mathcal{\hat{H}} \mid \hat{\rho}) = \left( \operatorname{Tr}\left((\mathcal{\hat{G}}_i \otimes \mathcal{\hat{H}}_j) \hat{\rho}_{\mathrm{diff}}\right) \right).
\end{equation}

Similarly, if Bob's measurements are untrusted, with $\mathrm{Tr}(\hat{\rho}_b \hat{B}_j) \leq 1$ and  $\Tr(\hat{A}_i\hat{A}_i^\dagger) \leq 1$}, the modified steering inequality becomes
\begin{align}
     & \|C(\mathcal{\hat{G}}, \mathcal{\hat{H}} \mid \hat{\rho})\|_{\mathrm{Tr}} - d^2\sqrt{\xi} \leqslant \nonumber                        \\
     & \sqrt{ \biggl(d_1 - \mathrm{Tr}\bigl(\hat{\rho}_a^2\bigr)\biggr)\biggl(1 - \mathrm{Tr}\bigl(\hat{\rho}_b^2\bigr) + 4d^2\eta\biggr)}.
    \label{2modyfy}
\end{align}}
Next, we will demonstrate through examples how small errors affect the steering inequality. Let us consider a asymmetric two-qubit state proposed in Ref. \cite{PhysRevA.106.042402}
\begin{equation}
    \hat\rho(p)=p|\psi^-\rangle\langle\psi^-|+\frac{1-p}{3}\left(2|0\rangle\langle0|\otimes\frac{\mathbb{I}}{2}+\frac{\mathbb{I}}{2}\otimes|1\rangle\langle1|\right),
\end{equation}
here $p \in[0,1]$ and $|\psi^-\rangle=\frac{1}{\sqrt{2}}(|01\rangle-|10\rangle)$ is the maximally entangled singlet state. If the measurements on Bob's side are trusted, i.e., deviations arising from device imperfections need not to be considered, and both parties perfectly implement the predefined measurement settings $(1, \sigma_x, \sigma_y, \sigma_z)/\sqrt{2}$, Lai and Luo  derived that $\rho(p)$ is steerable from Alice to Bob for $p>0.577$ and   steerable from Bob to Alice for $p>0.565$. Now if Bob's measurements have imprecision of $0.001\%$, then only when $p>0.716$, one can say the system is steerable from Alice to Bob, and it is steerable from Bob to Alice while $p>0.715$. In Fig. \ref{fig:tu1}, we plot the variation of parameter $p$ ensuring quantum steering as a function of the error $\xi$. However, in some specific numerical cases, our method does not outperform the previous scheme in Ref. \cite{PhysRevA.108.L040401} in terms of error tolerance and detection robustness. For example, for the asymmetric two-qubit state $\rho(p)$ analyzed in the text, under identical measurement errors, the steerability threshold $p$ derived from our modified inequality is generally looser than that in Ref. \cite{PhysRevA.108.L040401}, indicating a higher sensitivity to errors. This discrepancy may stem from the different ways the two methods handle error propagation: the derivation based on correlation matrices introduces additional terms during error transmission, which may lead to a more pronounced error amplification effect, making the corrected bound relatively conservative. Remarkably, our method still offers an independent and novel perspective for analyzing error influence from the viewpoint of correlation matrices, contributing to a more detailed understanding of the structural propagation of errors in high-dimensional systems.

\begin{figure}
    \begin{minipage}{0.45\textwidth}
        \centering

        \subfigure{\includegraphics[width=8cm]{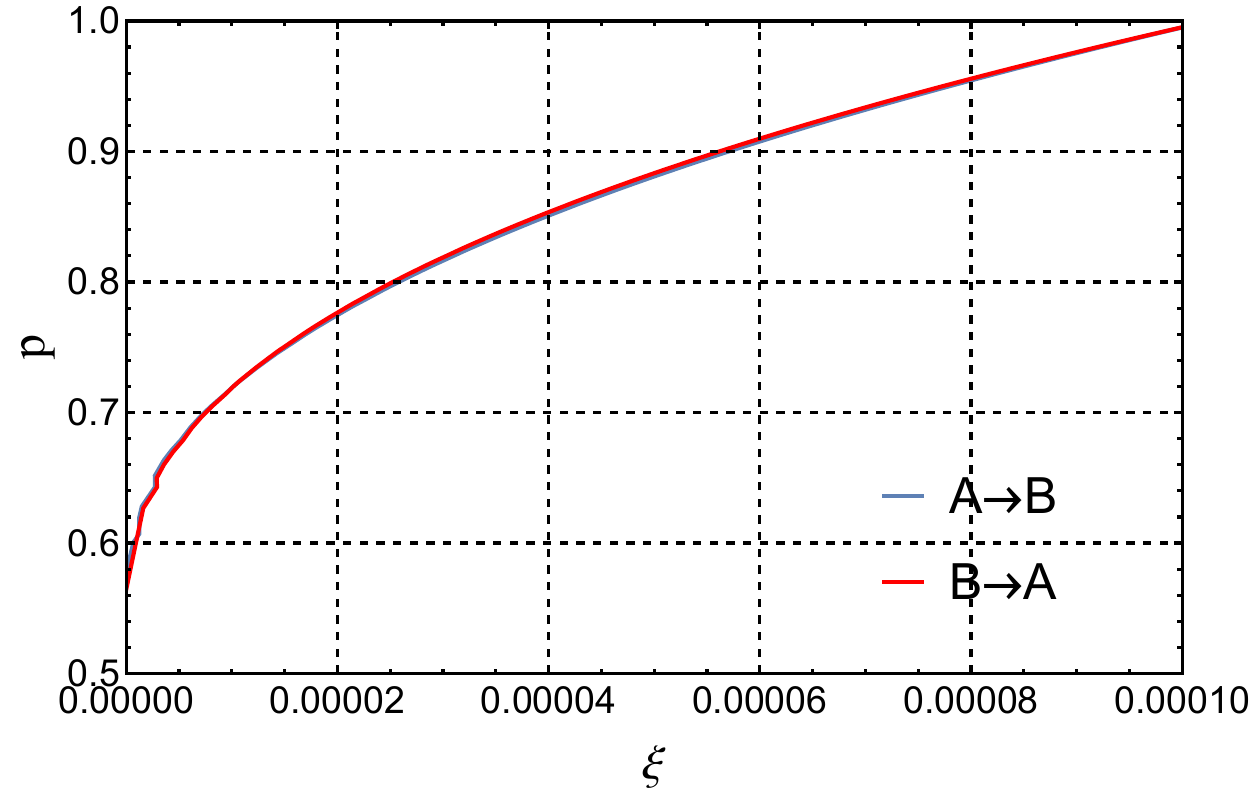}}
    \end{minipage}\hfill
    \caption{The dependence of the steerability parameter $p$ of a quantum state on the magnitude of experimental error, quantified in the range $\xi \in [0,10^{-4}]$. Here, $p$ represents the critical threshold for certifying quantum steering, with values $p \geq 0.577$ ($A\rightarrow{B}$) and $p \geq 0.565$ ($B\rightarrow{A}$)  indicating steerable states under ideal conditions. As the error $\xi$ increases from 0 to 0.01$\%$, the parameter $p$ monotonically increases  to 1, demonstrating a degradation of steerability with rising imperfections.}
    \label{fig:tu1}
\end{figure}

Quantum steering characterizes the nonlocal phenomenon where measurements performed on one subsystem of a bipartite quantum system instantaneously influence the state of the other subsystem. To investigate steering in a three-qubit system, it is necessary to partition the tripartite system into two distinct parties. Here, we adopt the bipartition, treating qubit A as an individual subsystem (\(\mathbb{C}^2\)) and combining qubits B and C into a joint subsystem (\(\mathbb{C}^2 \otimes \mathbb{C}^2 \simeq \mathbb{C}^4\)). This bipartition enables the representation of the composite system within the Hilbert space $\mathcal{H} =\mathbb{C}^2 \otimes \mathbb{C}^4$.

Besides, Li et al. derive a steering criterion for tripartite system in Ref. \cite{PhysRevA.110.012418}. For a tripartite quantum system partitioned into subsystems \(A\) and \(BC\), two sets of LOOs: \(G^A_p\) for subsystem \(A\) and \(G^{BC}_{q}\) for the joint subsystem \(BC\), are chosen to detect the steerability of the quantum state \(\hat\rho\). The LOOs are explicitly defined as:
\begin{equation}
    G^A_p = \frac{1}{\sqrt{2}} \, \hat{\sigma}_p \quad \text{for } p \in \{0, 1, 2, 3\},
\end{equation}
\begin{equation}
    G^{BC}_{q} = \frac{1}{2} \, \hat{\sigma}_j \otimes \hat{\sigma}_k \ \text{for } j,k \in \{0, 1, 2, 3\}, q \in \{0,1,...,15\}.
\end{equation}
Then, for an arbitrary tripartite quantum state \(\hat\rho_{ABC} = \frac{1}{8} \sum_{i,j,k=0}^3 \Theta_{ijk} \, \hat{\sigma}_i \otimes \hat{\sigma}_j \otimes \hat{\sigma}_k\), if the trace norm of the correlation matrix \(M\) satisfies:
\begin{equation}
    \| M \|_{\mathrm{Tr}} > \sqrt{\big(2 - \mathrm{Tr}(\hat{\rho}_A^2)\big)\big(1 - \mathrm{Tr}(\hat{\rho}_{BC}^2)\big)},
\end{equation}
then  \(\hat{\rho}_{ABC}\) is steerable from subsystem \(A\) to the joint subsystem \(BC\). Here, \(M\) is the correlation matrix constructed from
\(G_p^A\) and \(G_q^{BC}\), the two factors under the square roots correspond to the modified variance terms $\Lambda_A = 2 - \mathrm{Tr}(\hat{\rho}_A^2)$ and $\Lambda_{BC} = 1 - \mathrm{Tr}(\hat{\rho}_{BC}^2)$, respectively, which are derived from the correlation matrix approach introduced in Ref. \cite{PhysRevA.110.012418}, and $\Theta_{ijk} = \operatorname{Tr}(\hat\rho \, \hat\sigma_i \otimes \hat\sigma_j \otimes \hat\sigma_k)$, with matrix elements defined as:
\begin{equation}
    M_{pq} = \mathrm{Tr}\left[ \left( \hat{G}_p^A \otimes \hat{G}_q^{BC} \right) \big( \hat{\rho}_{ABC} - \hat{\rho}_A \otimes \hat{\rho}_{BC} \big) \right].
\end{equation}
Here, we can similarly treat the set \(\{G\}\) as a complete set of measurement bases. The discrepancy in measurement errors on Bob's side between the bipartite and tripartite scenarios stems purely from the difference in dimensionality. Specifically, the joint subsystem \(BC\) (spanning \(\mathbb{C}^4\)) introduces higher-dimensional noise propagation, whereas the mathematical structure of the error model remains analogous to the bipartite case.
\begin{proposition}\label{proposition2}
    With the above notation, if the state $\hat{\rho}$ on $\mathbb{C}^d \otimes \mathbb{C}^{d^2}$ is steerable from $A$ to $BC$, then
    \begin{equation}
        \| M \|_\mathrm{Tr} -d^3\sqrt{\xi} > \sqrt{ \biggl(d - \mathrm{Tr}(\hat{\rho}_A^2) \biggr) \biggl(1 - \mathrm{Tr}(\hat{\rho}_{BC}^2) + 4d^4\eta \biggr)}.
        \label{3modyfy}
    \end{equation}
\end{proposition}

In addition, leveraging the intrinsic asymmetry of quantum steering, we establish a criterion to verify whether the joint system \( BC \) can steer subsystem \( A \) \cite{PhysRevA.110.012418}. The protocol proceeds as follows: First, we apply a cyclic permutation operator to transform the tripartite state \( \hat\rho_{ABC} \) into \( \hat\rho_{BCA} \), explicitly given by
\begin{equation}
    \hat\rho_{BCA} = \frac{1}{8} \sum_{i,j,k=0}^{3} \Theta_{ijk} \, \hat{\sigma}_j \otimes \hat{\sigma}_k \otimes \hat{\sigma}_i.
\end{equation}
The steering criterion then states that \( \hat\rho_{ABC} \) is steerable from \( BC \) to \( A \) if the following inequality holds:
\begin{equation}
    \|M'\|_{\mathrm{tr}} > \sqrt{\biggl[4 - \mathrm{Tr}\bigl(\hat{\rho}_{BC}^2\bigr)\biggr] \biggl[1 - \mathrm{Tr}\bigl(\hat{\rho}_{A}^2\bigr)\biggr]},
\end{equation}
where \( M' \) is the correlation matrix constructed from LOOs \(\{G_q^{BC}\}\) and \(\{G_p^{A}\}\). The matrix elements are defined as:
\begin{align}
    M'_{qp} & = \mathrm{Tr}\left[\left(\hat{G}_q^{BC} \otimes \hat{G}_p^{A}\right)\left(\hat{\rho}_{BCA} - \hat{\rho}_{BC} \otimes \hat{\rho}_{A}\right)\right] \nonumber \\
            & = \Theta_{[q/4](q/4)p} - \Theta_{[q/4](q/4)0}\Theta_{00p}.
\end{align}

\begin{proposition}\label{proposition3}
    In the scenario where \( BC \) steers \( A \), the measurement imperfections are confined solely to the untrusted measurements on Alice's side. The modified steering inequality under such conditions is given by:
    \begin{equation}
        \|M'\|_{\mathrm{tr}} -d^3\sqrt{\xi} > \sqrt{ \biggl(d^2 - \mathrm{Tr}(\hat{\rho}_{BC}^2) \biggr) \biggl(1 - \mathrm{Tr}(\hat{\rho}_A^2) + 4d^2\eta \biggr) }. \label{555}
    \end{equation}
\end{proposition}
Here, proposition \ref{proposition2} and proposition \ref{proposition3} formulate steering criterion for tripartite systems with imprecise measurements, we thus introduce gap values defined as the difference between the left- and right-hand sides of Eqs. (\ref{2modyfy}), (\ref{3modyfy}) and (\ref{555})
\begin{equation}
    \begin{split}
        H_{A \rightarrow B}&=\|C(\mathcal{\hat{G}}, \mathcal{\hat{H}} \mid \hat{\rho})\|_{\mathrm{Tr}} - d^2\sqrt{\xi} \\
        &-  \sqrt{ \biggl(d_1 - \mathrm{Tr}\bigl(\hat{\rho}_a^2\bigr)\biggr)\biggl(1 - \mathrm{Tr}\bigl(\hat{\rho}_b^2\bigr) + 4d^2\eta\biggr)},
    \end{split}
\end{equation}
\begin{equation}
    \begin{split}
        H_{A \rightarrow BC}&=\| M \|_\mathrm{Tr} - d^3\sqrt{\xi} \\
        &-\sqrt{ \biggl(d - \mathrm{Tr}(\hat{\rho}_A^2) \biggr) \biggl(1 - \mathrm{Tr}(\hat{\rho}_{BC}^2) + 4d^4\eta \biggr)},
    \end{split}
\end{equation}
\begin{equation}
    \begin{split}
        H_{BC \rightarrow A}&=\|M'\|_{\mathrm{tr}} - d^3\sqrt{\xi} \\
        &- \sqrt{ \biggl(d^2 - \mathrm{Tr}(\hat{\rho}_{BC}^2) \biggr) \biggl(1 - \mathrm{Tr}(\hat{\rho}_A^2) + 4d^2\eta \biggr) }.
    \end{split}
\end{equation}
The presence of steering from $A$ to $B$, from $A$ to $BC$, and from $BC$ to $A$ is directly quantified by the values of $H_{A\rightarrow B}$, $H_{A\rightarrow BC}$, and $H_{BC\rightarrow A}$, respectively. Specifically, $H_{A\rightarrow B}>0$, $H_{A\rightarrow BC}>0$, or $H_{BC\rightarrow A}>0$ signify the emergence of quantum steering.

We first analyze a typical three-qubit quantum state, the generalized Greenberger-Horne-Zeilinger (GHZ) state, which is formally defined by the form of
\begin{equation}
    |\psi\rangle_{\text{GHZ}} = \sin\theta|000\rangle + \cos\theta|111\rangle,
    \label{eq:ghz_state}
\end{equation}
where $0 \leq \theta \leq \pi/2$. For this state, one can easily work out that
\begin{align}
     & \| M \|_\mathrm{Tr} = 2\cos\theta \sin\theta + 2\cos^2\theta \sin^2\theta, \label{eq:trace_norm}
    \\&2 - \mathrm{Tr}(\hat\rho_A^2) = \frac{1}{4}(5 - \cos 4\theta), \label{eq:state_purity}
    \\&1 - \mathrm{Tr}(\hat\rho_{BC}^2) = 2\cos^2\theta \sin^2\theta. \label{eq:reduced_state}
\end{align}
These expressions are then substituted into the modified steering inequality Eqs. (\ref{3modyfy}) to assess the robustness of steering under measurement imperfections.
To probe the effect of the error on steerability, we plot the quantum steerability $ H_{A \rightarrow BC}$ as functions of the coefficient $\theta$ and the error $\xi$ in Fig. \ref{fig:tu2}. From the figure, it can be clearly observed that as $\xi$ increases, the requirements on $\theta$ to ensure system steerability become more stringent. Notably,  the   system will be fully unsteerable when the error $\xi$ is beyond a critical threshold (e.g., $\xi>4.82\times 10^{-5}$ in Fig. \ref{fig:tu2}). Interestingly, it can be found that the curve is not smooth when the errors $\xi$ are relatively small. This is because of the trace norm within the formulas (\ref{3modyfy}), which  essentially  guarantees the strict validity of the modified inequality.

\begin{figure}
    \centering
    \subfigure{\includegraphics[width=9cm]{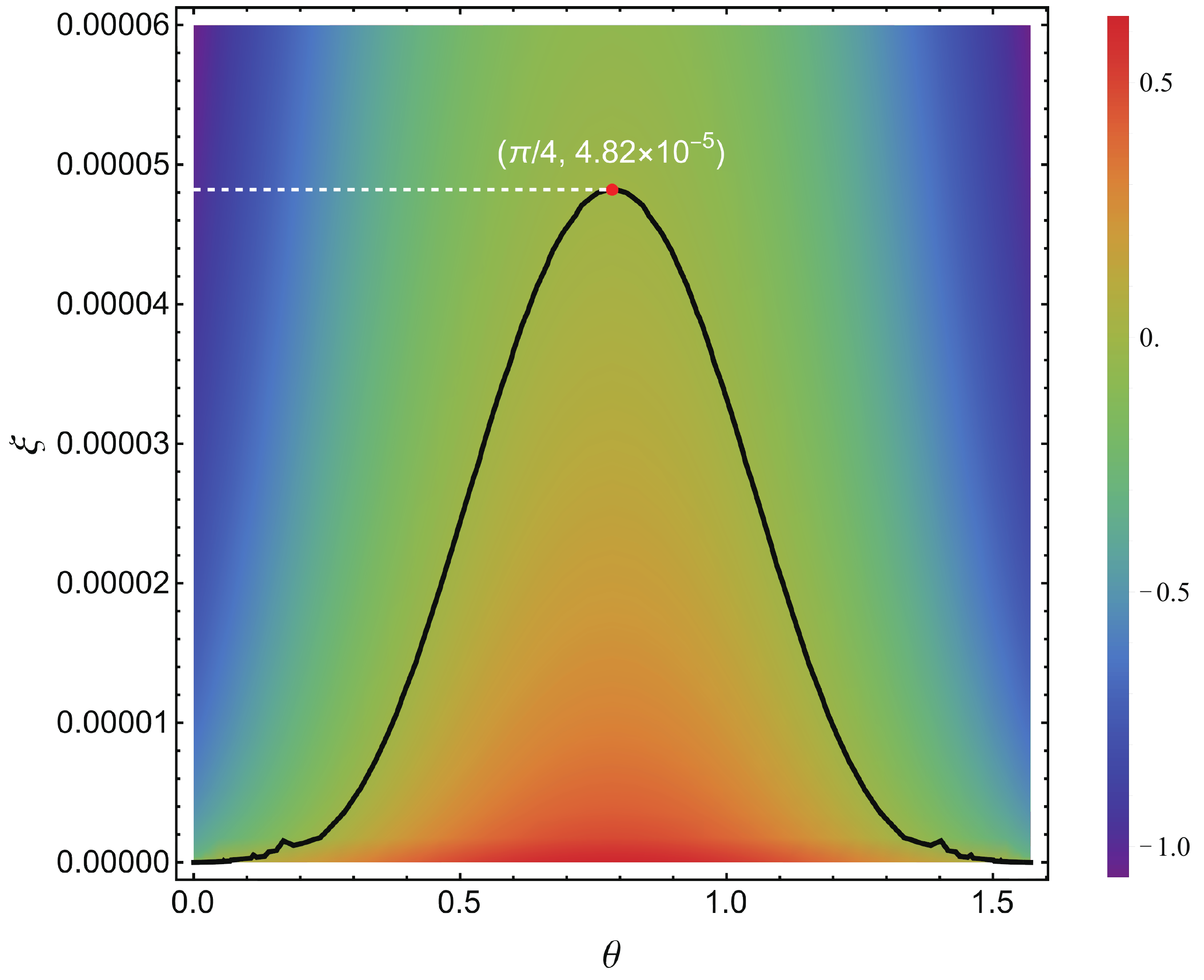}}
    \caption{The quantum steerability $ H_{A \rightarrow BC}$ with respect to measurement imperfection $\xi$ and the state's parameter $\theta$. The black solid line represents the zero-valued steering (i.e., $ H_{A \rightarrow BC}=0$). The black solid line denotes the critical steering boundary. To be explicit, the states  are steerable in the regions of below the boundary, contrarily, the systems are unsteerable in the regions of  above the boundary. When $\xi$ exceeds a critical value, the system becomes completely unsteerable which is independent of $\theta$.}
    \label{fig:tu2}
\end{figure}


To clearly demonstrate the impact of dimensionality on error sensitivity, we consider a $d$-dimensional maximally entangled state (generalized Bell state):
\begin{equation}
    |\Phi_d\rangle = \frac{1}{\sqrt{d}} \sum_{i=0}^{d-1} |ii\rangle,
\end{equation}
here we choose the generalized Pauli bases for measurements. According to our derived results (\ref{proof10}), (\ref{3modyfy}) and (\ref{555}), one can analytically obtain the influence of the error $\xi$ on steering inequality  under different dimensions $d$. To compare the dimension-dependent error sensitivity, we plot the  $H_{A \rightarrow B}, H_{A \rightarrow BC}$, and $H_{BC \rightarrow A}$ for different dimensions as a function of the error $\xi$ in Fig.~\ref{f1}.


\begin{figure*}
    \begin{minipage}{0.32\textwidth}
        \centering
        \includegraphics[width=6cm]{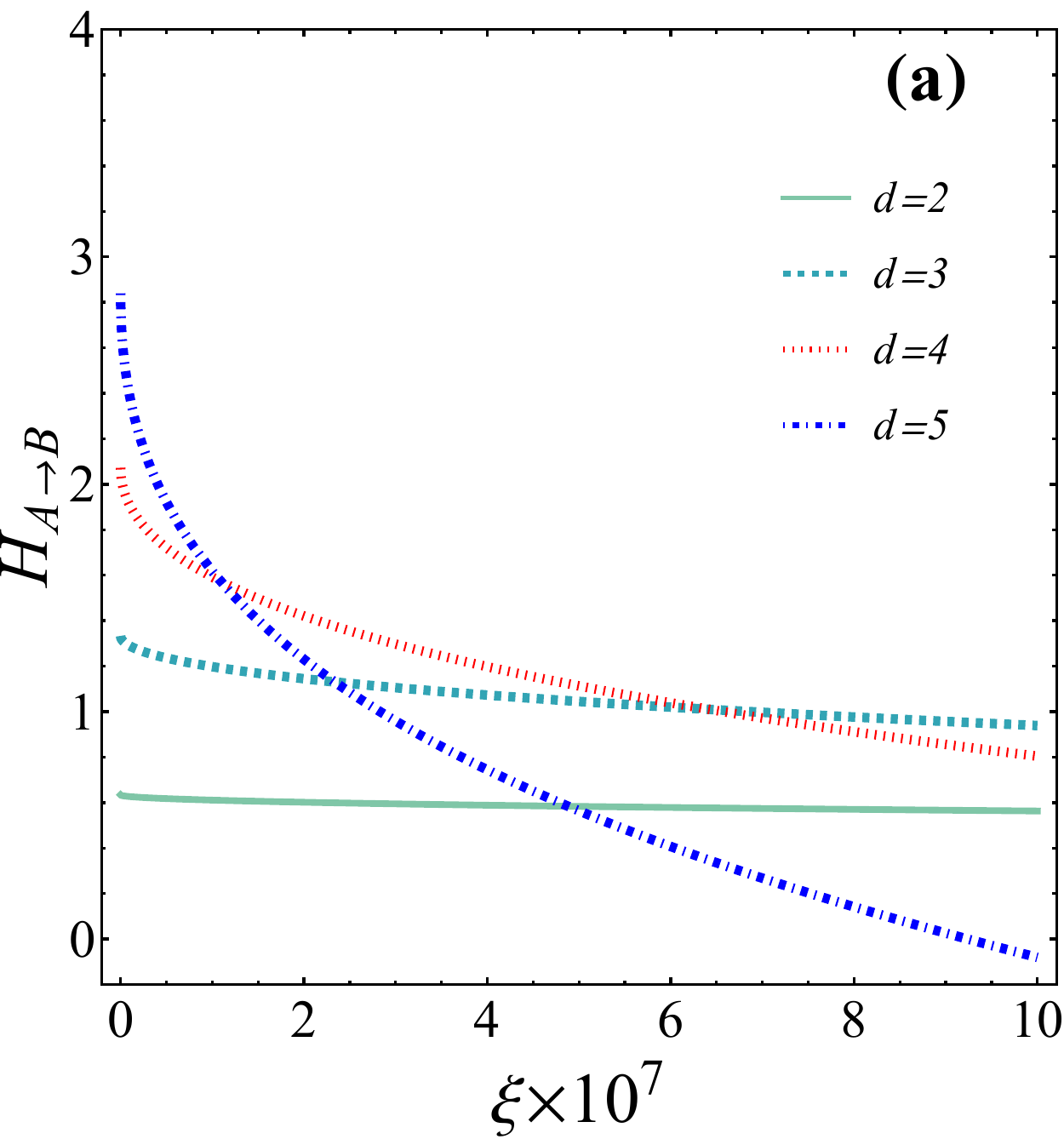} 
        \label{fig:side:a}
    \end{minipage}
    \hfill 
    \begin{minipage}{0.32\textwidth}
        \centering
        \includegraphics[width=6cm]{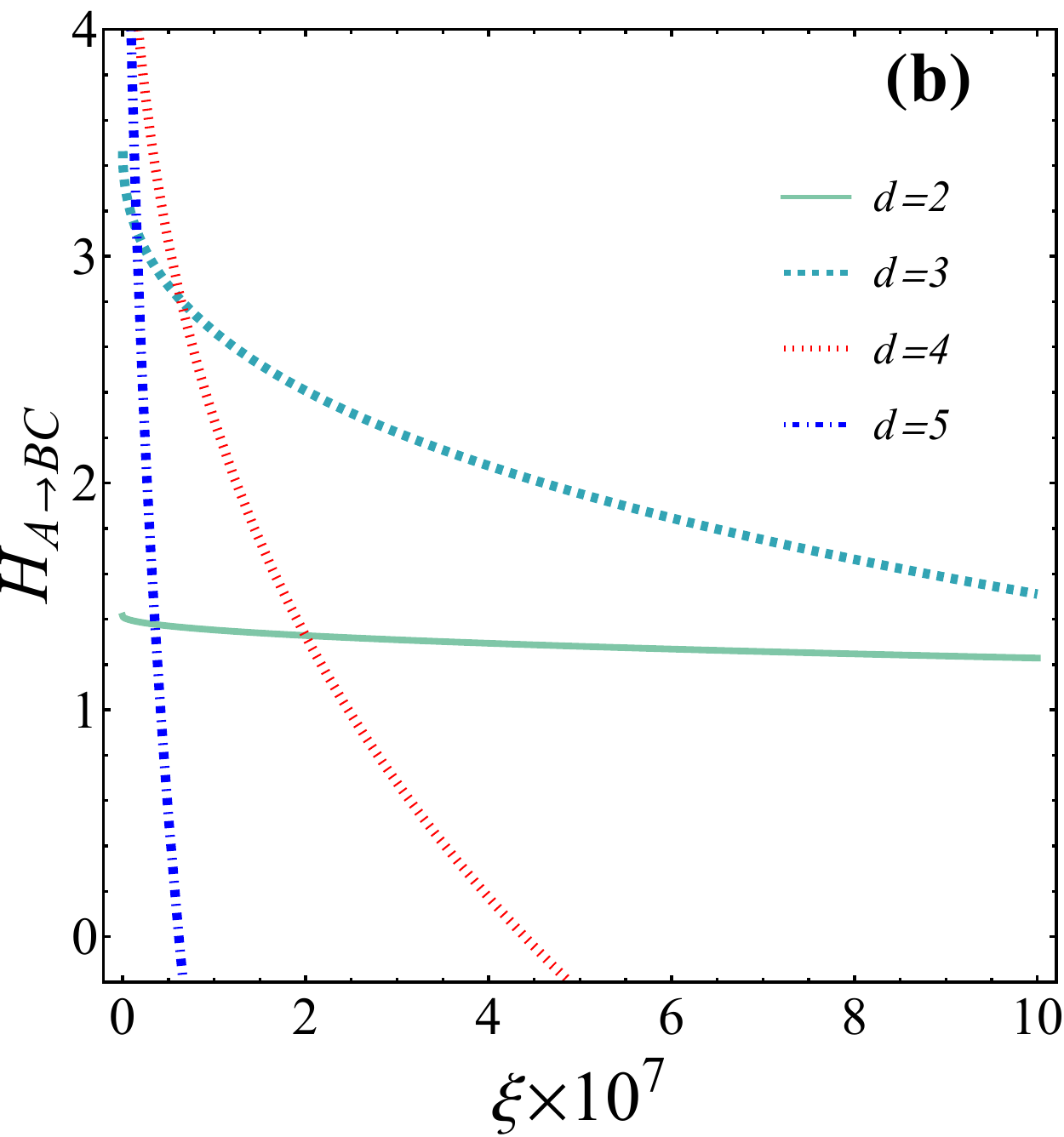} 
        \label{fig:side:b}
    \end{minipage}
    \hfill 
    \begin{minipage}{0.32\textwidth}
        \centering
        \includegraphics[width=6cm]{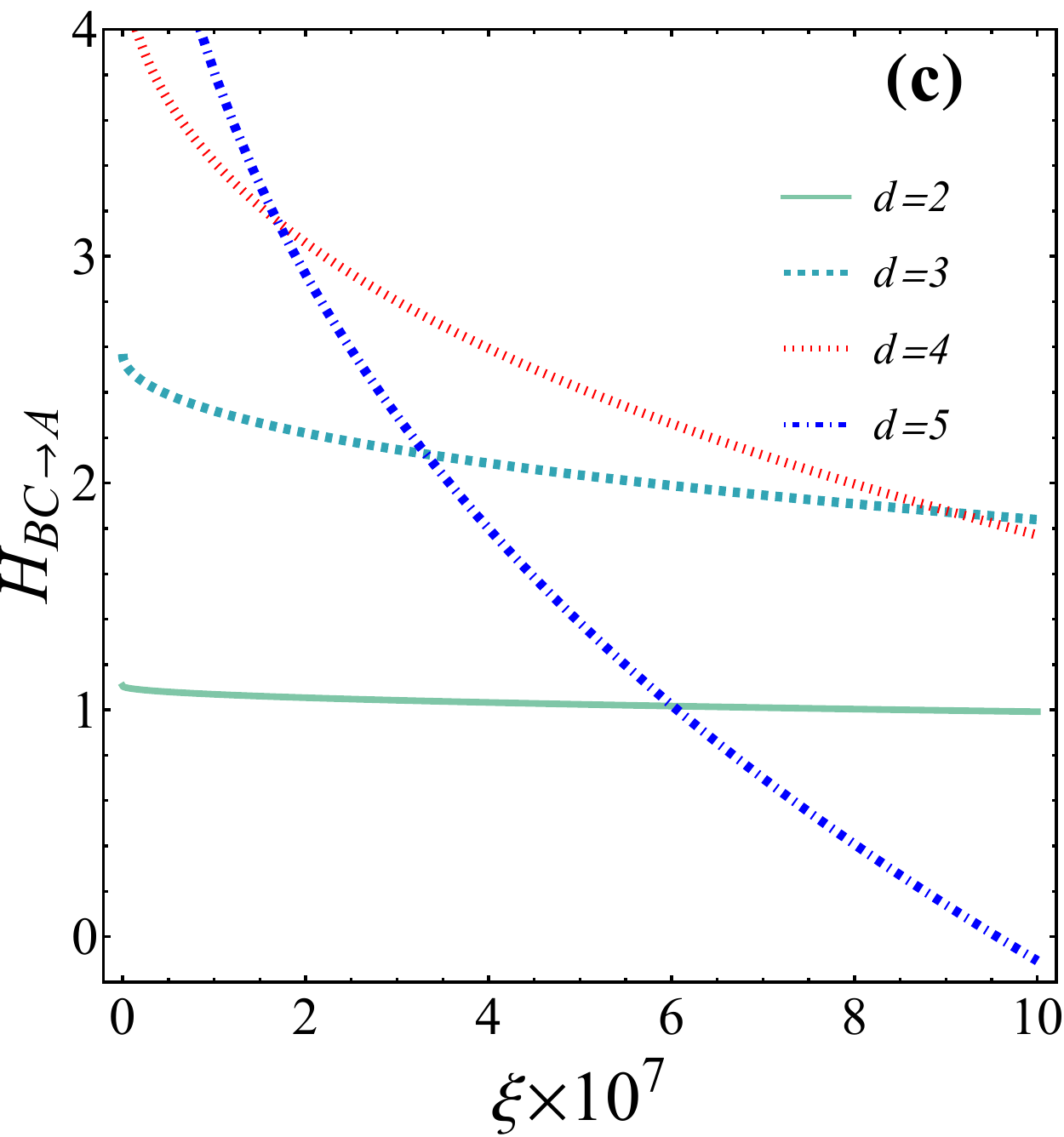} 
        \label{fig:side:c}
    \end{minipage}
    \caption{ Comparison of dimension-dependent error sensitivity in quantum steering protocols. (a) Bipartite systems ($A \rightarrow B$): Steering detectability degrades with increasing measurement error $\xi$, with higher-dimensional systems showing accelerated degradation.  (b) Tripartite $A \rightarrow BC$ steering: the joint subsystem $BC$ (effectively $d^2$-dimensional) introduces amplified noise propagation, making this configuration particularly vulnerable to measurement imperfections. (c) Tripartite $BC \rightarrow A$ steering: the error effects are confined to the single untrusted subsystem $A$, leading to  robustness against quantum imperfection in comparison with the case in Graph (b).}
    \label{f1}
\end{figure*}
In Fig.~\ref{f1} (a), for bipartite systems, the steering parameter $H_{A \rightarrow B}$ declines rapidly as either the error $\xi$ or the dimension $d$ increases, indicating that even small measurement inaccuracies can severely compromise steerability in high-dimensional systems and underscoring the exponentially growing demand for measurement precision. Fig.~\ref{f1} (b) depicts the tripartite case where subsystem $A$ steers the joint subsystem $BC$: $H_{A \rightarrow BC}$ exhibits an exceptionally sharp drop with the increasing $\xi$ or $d$. This is because of the effectively $d^2$-dimensional nature of $BC$ and the fact that both $B$ and $C$ are untrusted, the errors are amplified quadratically, resulting in the significant reduction of $H_{A \rightarrow BC}$ with the increasing $\xi$ or $d$. In contrast, Fig.~\ref{f1} (c) shows the case from $BC$ to $A$, where measurement imperfections are confined solely to the single untrusted party $A$; here, the decrease in $H_{BC \rightarrow A}$ is noticeably more gradual, reflecting substantially enhanced robustness due to restricted error propagation. We therefore conclude that the dimension $d$ constitutes a critical determinant of measurement error sensitivity in the multipartite quantum steering protocols. Moreover, our analysis indicates that, for multipartite steering protocols,  a reduced number of untrusted parties enhances robustness against measurement imperfections.

In tripartite systems, the impact of measurement inaccuracies on steering inequalities was analyzed under two scenarios: one party steers a joint two-party subsystem and two parties jointly steer a single party. For the former case, the steering inequality involves higher-dimensional measurements (e.g., $BC$ as a $d^2$-dimensional system), where errors scale with dimensionality, amplifying noise sensitivity. For the latter case, where $BC$ steers $A$, errors only affect the single-party subsystem ($A$), leading to reduced error propagation. Steering from two parties to one exhibits superior robustness against measurement errors compared to steering from one party to two, due to lower-dimensional error propagation and reduced cumulative noise effects in the former case. This asymmetry highlights the importance of error-aware partitioning in multipartite steering protocols.
\section{Discussion and Conclusion}\label{Sec5}

{Multipartite quantum steering serves as a foundational resource for advancing quantum communication and computation within distributed quantum networks. In this work, we systematically investigated the impact of measurement inaccuracies on steering detection, particularly focusing on arbitrary dimensional bipartite systems and tripartite systems. Our primary contribution lies in deriving robust steering criteria for correlation matrices under imperfect measurements, demonstrating that even minor deviations in trusted devices such as Bob's measurement operators significantly degrade steering verification, with error sensitivity scaling for \(d\)-dimensional systems, reference \cite{PhysRevA.108.L040401, PhysRevLett.132.070204} also notes this point. We further extended these criteria to tripartite scenarios, establishing modified inequalities for bipartition systems, and discovered that in multipartite steering, under the same error level, the number of steered parties significantly affects measurement accuracy, and this effect combines with dimensional amplification, thereby further reinforcing the amplified noise propagation in high‑dimensional subspaces.

    Notably, we proved that the steering thresholds based on trace norm are fundamentally altered by measurement errors, as exemplified by the asymmetric two-qubit state \(\hat\rho(p)\), where a \(0.01\%\) error elevates the steerability threshold from \(p > 0.577\) to \(p > 0.766\). For tripartite systems, we formulated a cyclic permutation protocol to analyze steering from \(BC\) to \(A\), revealing that errors confined to untrusted measurements on \(A\)'s side introduce additive noise terms proportional to \(\eta\) in the modified inequality.

    To validate our theoretical framework, we employed numerical and analytical case studies, including asymmetric entangled states , tripartite GHZ-type states and generalized Bell state, which highlighted the operational resilience and limitations of our criteria under realistic noise conditions. These results underscore the necessity of near-perfect device characterization in scalable quantum architectures and provide a semi-device-independent methodology for error quantification.

    Given our findings on the pronounced effects of measurement inaccuracies in high-dimensional and multipartite steering, we note that an emerging research direction, the “imprecision plateau”, offers a promising approach offering the potential to address the issue of imprecise measurement \cite{PhysRevA.111.L020404}. Recent studies have shown that by carefully designing steering inequalities, it is possible to completely avoid degradation in detection capability when the trusted party’s measurements are subject to limited imperfection, provided the inaccuracy remains below a critical threshold  \cite{PhysRevA.111.L020404}. This implies that even when relaxing the unrealistic assumption of perfect measurements, the noise tolerance and detection efficiency thresholds of the experiment can be preserved, and performance in applications such as semi-device-independent randomness generation can be maintained. This phenomenon reveals the potential intrinsic robustness of quantum steering tests against measurement errors. Consequently, future work could aim to develop steering criteria that are both applicable to high-dimensional systems and robust to realistic noise, exhibiting a long “imprecision plateau”. This shifts the focus from analyzing error effects to designing reliable, efficient certification protocols, thereby laying a stronger foundation for practical quantum networks.

    To sum up, our findings not only advance the understanding of steering robustness in imperfect experimental settings but also pave the way for future research into multipartite steering criteria and error-mitigation strategies for high-dimensional quantum resources.
    This work establishes a critical foundation for reliable quantum certification in practical, noise-prone quantum networks.}
\section{Acknowledgments}
This work was supported by the National Science Foundation of China (Grant nos. 12475009, 12075001, and 62471001), Anhui Provincial Key Research and Development Plan (Grant No. 2022b13020004), Anhui Province Science and Technology Innovation Project (Grant No. 202423r06050004),  Anhui Province Natural Science Foundation (Grant no. 202508140141), Anhui Provincial Department of Industry and Information Technology (Grant No. JB24044), and Anhui Provincial University Scientific Research Major Project (Grant No. 2024AH040008).

\bibliography{cite}

\end{document}